\documentclass[lettersize,journal]{IEEEtran}
\usepackage{amsmath,amsfonts}
\usepackage{algorithmic}
\usepackage{array}
\usepackage{subcaption}
\usepackage{balance}
\usepackage[hyphens]{url}
\usepackage{xurl}

\usepackage{textcomp}
\usepackage{stfloats}
\usepackage{url}
\usepackage{verbatim}
\usepackage{graphicx}
\usepackage[linesnumbered,ruled,vlined]{algorithm2e}
\usepackage{xcolor}
\usepackage{amsthm}
\usepackage{colortbl}

\newtheorem{theorem}{Theorem}
\newtheorem{lemma}[theorem]{Lemma}
\newtheorem{corollary}[theorem]{Corollary}

\newtheorem{definition}{Definition}[section]
\def\BibTeX{{\rm B\kern-.05em{\sc i\kern-.025em b}\kern-.08em
    T\kern-.1667em\lower.7ex\hbox{E}\kern-.125emX}}
\usepackage{balance}
\usepackage[caption=false,font=normalsize,labelfont=sf,textfont=sf]{subfig}
\begin{document}
\title{Decentralization in PoS Blockchain Consensus: Quantification and Advancement}

\author{%
    \IEEEauthorblockN{Shashank Motepalli and Hans-Arno Jacobsen}\\
    \IEEEauthorblockA{Department of Electrical and Computer Engineering \\ University of Toronto, Canada }
}


\maketitle

\begin{abstract}
Decentralization is a foundational principle of permissionless blockchains, with consensus mechanisms serving a critical role in its realization. This study quantifies the decentralization of consensus mechanisms in proof-of-stake (PoS) blockchains using a comprehensive set of metrics, including Nakamoto coefficients, Gini, Herfindahl-Hirschman Index (HHI), Shapley values, and Zipf’s coefficient. Our empirical analysis across ten prominent blockchains reveals significant concentration of stake among a few validators, posing challenges to fair consensus. To address this, we introduce two alternative weighting models for PoS consensus: Square Root Stake Weight (SRSW) and Logarithmic Stake Weight (LSW), which adjust validator influence through non-linear transformations. Results demonstrate that SRSW and LSW models improve decentralization metrics by an average of 51\% and 132\%, respectively, supporting more equitable and resilient blockchain systems.
\end{abstract}


\section{Introduction}
Bitcoin shaped the field of blockchains by introducing a peer-to-peer system that operates without trusted intermediaries~\cite{nakamoto2008bitcoin}. Its vision encapsulates \textit{decentralization}, characterized by the elimination of single points of failure. Our work focuses on exploring decentralization in blockchains, particularly in their consensus mechanisms. These mechanisms are critical as they establish agreement on the content and order of transactions among validators.

The problem this paper addresses is the \textit{analysis and advancement of decentralization in consensus mechanisms}. This problem is particularly interesting because, although decentralization is fundamental to every blockchain, standardized metrics to quantify it within consensus are lacking. This gap, coupled with the technical complexity of consensus algorithms, complicates the analysis of blockchain systems in practice. Moreover, enhancing decentralization in consensus mechanisms is of significance, as it directly contributes to the trust and safety in blockchain systems.

To address the challenge of decentralization in proof-of-stake (PoS) blockchains, we analyze how a validator's stake influences their voting power in weighted consensus mechanisms. Drawing inspiration from quantitative decentralization metrics originally developed for decentralized autonomous organizations (DAOs) with token-based voting~\cite{austgen2023dao,sharma2023unpacking,tan2023open}, we extend these frameworks to measure decentralization within consensus mechanisms. Specifically, we adapt metrics such as the Gini index and Nakamoto coefficients for both safety and liveness to assess validator influence~\cite{motepalli2024does}. Our empirical analysis spans ten major blockchains: Aptos~\cite{aptos}, Axelar~\cite{axelar}, BNB~\cite{bnb}, Celestia~\cite{celestia}, Celo~\cite{celo}, Cosmos~\cite{cosmos}, Injective~\cite{injective}, Osmosis~\cite{osmosis}, Polygon~\cite{polygon}, and Sui~\cite{sui}. Findings reveal a concentration of voting power among a small subset of validators, raising critical concerns for blockchain security and integrity. To address this, we propose the Square Root Stake Weight (SRSW) function, demonstrating its effectiveness through extensive evaluation~\cite{motepalli2024does}.

This paper extends our \textcolor{black}{previous work~\cite{motepalli2024does}} across four main dimensions. First, we broaden the analysis of decentralization by adapting four new interdisciplinary metrics, including the Shapley value from cooperative game theory to capture coalition influence, along with the Herfindahl-Hirschman Index and Zipf’s law coefficient. Second, we introduce the Logarithmic Stake Weight (LSW) model to further decentralize consensus mechanisms, showing that LSW consistently outperforms SRSW and current PoS models by an average of 51\% and 132\% across all metrics, respectively. \textcolor{black}{Third, we provide empirical assessments for both SRSW and LSW models} and extending the formal analysis with rigorous proofs that justify the observed results. Lastly, we broaden the scope of our empirical analysis by incorporating additional metrics across our blockchain set, using the latest available data to ensure relevance.

Our contributions are four-fold: 
\vspace{-2pt}
\begin{enumerate}
    \item Adapt and formalize decentralization metrics specifically for evaluating consensus mechanisms.
    \item We empirically demonstrate the challenge of weight concentration and its impact on decentralization in prominent blockchains. 
    \item We introduce the Square Root Stake Weight (SRSW) and Logarithmic Stake Weight (LSW) as mechanisms to address the weight concentration challenge. 
    \item Our work pioneers a data-driven approach to consensus research, offering novel insights and solutions to enhance decentralization in blockchain systems.
\end{enumerate}

This paper is organized as follows. We begin with a classification of consensus mechanisms, based on finality, in the following section. After adapting metrics to capture decentralization in Section~\ref{sec:metrics}, we perform an extensive analysis of empirical data in Section~\ref{sec:empiricalanalysis}. To address the stake concentration challenges identified in our analysis, we introduce the SRSW and LSW models in Section~\ref{sec:srsw-model}. In Section~\ref{sec:evaluation}, we evaluate the effectiveness of our models both theoretically and empirically. The paper concludes with a discussion of related work in Section~\ref{sec:relatedwork} and identifies avenues for future research in Section~\ref{sec:conclusion}.

\section{Consensus Mechanism Foundations and Classification}
\label{sec:consensusclassification}
This section delineates the consensus mechanisms and validator set selection process, setting the stage for the empirical analysis of decentralization in blockchains.

\subsection{Classification of Consensus Mechanisms Based on Finality}
A blockchain comprises a ledger that has blocks of transactions. The consensus mechanism orchestrates the process of reaching agreement on the content and order of blocks within the ledger~\cite{zhang2022reaching}. Agreement is reached among a designated set of validators, sometimes referred to as miners~\cite{nakamoto2008bitcoin}, witnesses~\cite{li2020comparison}, or sequencers~\cite{motepalli2023sok}. A consensus mechanism is deemed Byzantine Fault Tolerant (BFT) if it withstands a certain proportion of validators with malicious behaviour, in addition to crash failures~\cite{lamport2019byzantine}. 

Two properties are guaranteed by consensus mechanisms: \textit{safety}, ensuring all correct validators agree on the same content and order of blocks, and \textit{liveness}, ensuring the continual production of new blocks without indefinite delays~\cite{pass2017analysis,garay2015bitcoin}. Building upon the safety property, we introduce the concept of finality, also known as commitment. Finality of a block \( b \) at time $t$, denoted \( 0 \leq f(b, t) \leq 1\), indicates the probability with which the block has been appended to the ledger~\cite{anceaume2020finality}. When \( f(b, t) =1 \), it signifies \textit{total finality}, i.e., the block $b$ cannot be reverted or abandoned. Achieving this level of finality is essential for the immutability of the ledger. There are two distinct ways to realize finality in consensus mechanisms:

\begin{definition}
\label{def:absolute-finality}
\textit{Instant Absolute Finality} is achieved when a block \( b \) is appended to the ledger at time $t_0$ and becomes irreversible instantly, such that \( f(b, t) = 1\quad\text{for all } t > t_0 \).
\end{definition}

\begin{definition}
\label{def:probabilistic-finality}
\textit{Eventual Probabilistic Finality} is achieved when the finality of an appended block \( b \) at time \( t_0 \) is expressed as \( f(b, t_0) = 1-\gamma \), where \( \gamma \) is a non-negative value less than 1 (i.e., \( 0 \leq \gamma < 1\)) that represents the deviation from total finality. Specifically, for any two points in time \( t_1 \) and \( t_2 \) such that \( t_2 > t_1 \), it follows that \( f(b, t_1) < f(b, t_2) \). As time progresses, \( f(b, t) \) gradually converges to 1 as \( \gamma \) approaches zero. 
\begin{equation}
\lim_{\gamma \to 0} f(b, t) = 1
\end{equation}
\end{definition}

Based on how consensus mechanisms achieve finality, they can be categorized into two types: Nakamoto-style and classical consensus, as shown in Table~\ref{tab:types-consensus-mechanisms}.
\begin{table}[]
\caption{Consensus mechanisms classified based on finality}
\label{tab:types-consensus-mechanisms}
\renewcommand{\arraystretch}{1.5}
\centering
\begin{tabular}{|l|l|l|}
\hline
\multicolumn{1}{|c|}{} & \multicolumn{1}{c|}{\textbf{Classical consensus}}                     & \multicolumn{1}{c|}{\textbf{Nakamoto-style}}                          \\ \hline \hline
\textbf{Finality}      & \begin{tabular}[c]{@{}l@{}}Absolute and \\ instant\end{tabular} & \begin{tabular}[c]{@{}l@{}}Probabilistic and \\ eventual\end{tabular} \\ \hline
\textbf{Priniciple}    & Safety over liveness                                                  & Liveness over safety                                                  \\ \hline
\textbf{Attestation}   & Quorum of validators                                                  & Only the proposer                                                     \\ \hline
\textbf{Resources}     & Requires a priori knowledge                                            & No constraints                                                        \\ \hline
\textbf{Communication} & Supports partial synchrony                                            & Synchronous                                                           \\ \hline
\textbf{Examples}      & \begin{tabular}[c]{@{}l@{}}PBFT~\cite{castro1999practical}, HotStuff~\cite{yin2019hotstuff}, \\ Tendermint~\cite{buchman2018latest,buchman2016tendermint}\end{tabular} & \begin{tabular}[c]{@{}l@{}}Nakamoto~\cite{nakamoto2008bitcoin},  \\Ouroborus~\cite{david2018ouroboros}\end{tabular}        \\ \hline
\textbf{In practise}   & Diem~\cite{baudet2019state}                                                                  & Bitcoin~\cite{nakamoto2008bitcoin}                                                               \\ \hline
\end{tabular}

\end{table}

The Nakamoto-style consensus embodies probabilistic finality, meaning that the system eventually approaches total finality with time~\cite{kim2023taxonomic}. This style is used in Bitcoin, where a block is considered to have reached finality after the confirmation of 6 subsequent blocks, approximately an hour~\cite{nakamoto2008bitcoin}. Nakamoto-style consensus mechanisms prioritize liveness over safety, ensuring the continual production of new blocks; however, the ledger's order remains susceptible to forking~\cite{garay2015bitcoin},  i.e., the current order of the ledger may be altered, until time \( t \).

Conversely, classical consensus mechanisms prioritize safety over liveness. In this style, no blocks are appended to the ledger until absolute finality is achieved, rendering finality deterministic and immediate~\cite{kim2023taxonomic}. An example is the PBFT consensus~\cite{castro1999practical}, where a designated proposer, one of the validators, broadcasts a block and, absolute finality is achieved when a quorum of validators attests on the proposed block and the block is appended to the ledger.

A notable distinction between these styles also lies in the block generation process. In Nakamoto-style consensus, a single proposer is responsible for proposing a block. This design \textcolor{black}{does not} presuppose knowledge of resources, such as the hash power in PoW, and assumes synchronous communication, i.e., messages are broadcast within a bounded time~\cite{lewis2021does}. On the other hand, classical consensus protocols assume a priori knowledge of the total available resources~\cite{buchman2016tendermint}, such as the stake distribution of validators in PoS. The notion of a quorum $\mathbb{Q}$, facilitated through certificates of attestation, necessitates having finite resources, a topic explored further in the subsequent subsection.

For the scope of this work, our focus is classical consensus mechanisms for multiple reasons. Primarily, these mechanisms facilitate fast finality along with high performance in terms of throughput and latency, compared to Nakamoto-style consensus~\cite{yin2019hotstuff,buchman2016tendermint}. Secondly, the employment of quorum $\mathbb{Q}$ certificates enables these protocols to function effectively in a partially synchronous environment, thereby tolerating indefinite periods of asynchrony~\cite{lewis2021does}. Thirdly, classical consensus mechanisms have undergone rigorous examination over several decades~\cite{zhang2022reaching}, with seminal contributions such as Raft~\cite{ongaro2014search} and PBFT~\cite{castro1999practical}, finding applications in safety-critical domains such as aviation systems~\cite{wensley1978sift,siewiorek2005fault}.

In the subsequent sections, the discussion extends to Sybil resistance and the intricacies of weighted classical consensus.

\subsection{Weighted Consensus}
Traditional classical consensus mechanisms, such as PBFT~\cite{castro1999practical}, HotStuff~\cite{yin2019hotstuff}, and PrestigeBFT~\cite{zhang2023prestigebft}, are designed to be able to tolerate up to one-third of the validator set being faulty, where a faulty validator may exhibit malicious behavior or be offline. In these mechanisms, a designated block proposer is required to collect attestations from the validator set to form a quorum certificate. Let the validator set be represented as \( N = \{n_1, n_2, \ldots, n_m \} \), where $n_i$ represents a validator. A quorum certificate is formed with attestations from at least a super-majority of validators, denoted as \( \mathbb{Q} \), such that:

\begin{equation}
\mathbb{Q} \geq \left(\frac{2}{3}\right) m \text{, where } m = |N|
\end{equation}
Note that while we assume the protocol can be able to tolerate up to one-third of the total validators being faulty, some protocols may have different failure assumptions~\cite{miller2016honey}. In such cases, $\mathbb{Q}$ must be adjusted accordingly.

Classical BFT consensus mechanisms were initially conceived for permissioned systems, where the identities of all validators are established. When deployed in permissionless environments with (pseudo)anonymous identities, these mechanisms become susceptible to Sybil attacks, wherein a malicious actor could create multiple validators to subvert the consensus mechanism~\cite{douceur2002sybil,motepalli2024delay}. To mitigate this vulnerability and achieve Sybil resistance without relying on trusted intermediaries, Algorand~\cite{gilad2017algorand}, Ouroboros~\cite{david2018ouroboros} and Tendermint~\cite{buchman2016tendermint} pioneered PoS mechanism. In PoS, validators stake the native tokens of the system as a means of establishing their identity. These tokens are subject to penalization if validators engage in malicious behavior~\cite{motepalli2021reward}. Given that the native tokens are finite and the security of the consensus impacts the tokens' market value, validators are rationally incentivized to act correctly, thus enhancing the system's security. The development of PoS protocols, characterized by variably staked tokens, paves the way for the adoption of weighted consensus.

\textbf{\textit{Weighted consensus}} encompasses traditional classical consensus as a subset, where traditional models are effectively a special instance with uniform weights across validators. In weighted consensus, validators have varying weights in the consensus mechanism~\cite{edwardThesis}. In practice, in PoS/DPoS blockchains such as Cosmos~\cite{buchman2016tendermint}, the influence of a validator \( n_k \) in the consensus is quantified by their weight, \( w_k > 0\), which is a function of their staked tokens \( s_k\).

\begin{equation}
    w_k = s_k
\vspace{-6pt}
\end{equation}
Unlike traditional classical consensus mechanisms where a quorum ceritifcate is achieved based on the absolute number of validators, in weighted consensus, a quorum necessitates garnering two-thirds of the total weight. The total weight \( W \) of the system is:
\vspace{-4pt}
\begin{equation}
    W = \sum_{k}^{m} w_k \quad \forall k \in N
    \vspace{-6pt}
\end{equation}
The quorum certificate for weighted consensus for validator set $N$, is given below:
\vspace{-4pt}
\begin{equation}
\mathbb{Q} \geq \left( \frac{2}{3}\right) W
\vspace{-4pt}
\end{equation}
 Moreover, higher weight could also imply higher rewards or a higher probability of being selected as a block proposer~\cite{david2018ouroboros}. The evolution from traditional to weighted consensus, underscored by PoS/DPoS, is a pivotal adaptation to suit permissionless blockchains. We focus on weighted consensus in the rest of this work.

\subsection{Validator Set Selection}
Classical weighted consensus assumes that resources, such as total staked tokens in PoS, are finite and known a priori. These staked tokens are used to rank candidates interested in becoming validators and to choose the validator set. Various mechanisms exist for validator set selection, including PoS~\cite{gilad2017algorand}, DPoS~\cite{saad2020comparative}, delay towers~\cite{motepalli2022decentralizing}, and reputation mechanisms~\cite{de2018pbft}. This work does not make specific assumptions regarding the mechanism of validator set selection; instead, it focuses on the validators' engagement in the consensus mechanism.

The validator set typically remains fixed for a specified time interval, known as an \textit{epoch}. Following each epoch, a new validator set is selected through a \textit{reconfiguration} event~\cite{duan2022foundations,motepalli2022decentralizing}. Reconfiguration tends to consider updated stakes and involves eliminating faulty validators. 

It is also important to acknowledge that some blockchains, such as Algorand~\cite{gilad2017algorand}, use mechanisms like random sortation to randomly select a subset of candidates as validators every epoch. Our study focuses on systems where the validator set is deterministically defined, thereby excluding blockchains that employ random committee selection processes.

\section{Consensus Decentralization Metrics}
\label{sec:metrics}
In consensus mechanisms, \textit{decentralization} means reaching an agreement on the contents and order of transactions without centralized control, ensuring that no single validator or group of validators dominates the process. While challenging to precisely define~\cite{kiayias2022sok,schneider2003decentralization,sharma2023unpacking}, decentralization is essential for consensus mechanisms, as it underpins trust in the blockchains.

Our discussion draws inspiration from the \( (m,\varepsilon,\delta) \)-decentralization model described in ``Impossibility of Full Decentralization in Permissionless Blockchains''~\cite{kwon2019impossibility}. Here, \( m \) indicates the cardinality of the validator set, 
and \( \varepsilon \) represents the weight disparity between the most influential (richest) 
and the \( \delta \)-th percentile validator. The ideal case is full decentralization, expressed as \( (m,0,0) \) for a sufficiently large \( m \), that occurs when all validators have equal influence. While the \( (m,\varepsilon,\delta) \)-model captures the essence of decentralization, 
it lacks quantifiable metrics for comparing the decentralization of different blockchains. Therefore, we introduce additional metrics, as shown in Table~\ref{tab:metrics_symbols}, to effectively quantify and compare decentralization across different blockchains.

\subsection{Validator Set Cardinality (\( m \))} The number of validators participating in the consensus mechanism, i.e., \( m = |N| \).

Inference: A higher \( m \) suggests better decentralization, aligning with the \( (m,\varepsilon,\delta) \)-model.

Limitations: In weighted consensus, \( m \) alone may not reflect true decentralization. For example, if \( m=1000 \) but one validator holds 90\% of the total weight, it contradicts the decentralization ethos.

\subsection{Gini Coefficient (\( G \))}
Gini measures wealth inequality, commonly used in socioeconomic studies~\cite{ceriani2012origins,gini1921measurement,sitthiyot2020simple}. In consensus mechanisms, it assesses validators' influence disparity, indicating deviation from \( (m,0,0) \)-decentralization.

\( G \) is calculated using the Lorenz curve, which graphically elucidates the weight distribution among validators~\cite{gastwirth1972estimation}. The Lorenz curve plots the cumulative share of validators (sorted by their weight) on the X-axis against the cumulative share of their weight on the Y-axis. The formula for \( G \) is:
   \vspace{-4pt}
    \begin{equation}
        G = 1 - \frac{2 \times B}{A + B}
    \end{equation}
where \( B \) is the area between the Line of Equality and the Lorenz Curve, and \( A \) is the area beneath the Lorenz Curve. The line of equality illustrates a hypothetical scenario of equal weight distribution, while the Lorenz curve depicts the actual distribution of weights~\cite{sitthiyot2021simple}. The area between these two curves represents the extent of inequality in the weight distribution~\cite{gastwirth1972estimation}.
        
Inference: \( G \) ranges between 0 and 1, with 0 indicating equitable distribution (higher decentralization) and values closer to 1 indicating concentration of weight (lower decentralization).

Limitations: \( G \) alone may not fully capture decentralization, as it does not account for validator set cardinality \( m \). For example, a system with a single validator (\( m =1\)) would have \( G =0 \), yet be highly centralized.
\begin{table}[]
\renewcommand{\arraystretch}{1.5}
\centering

\caption{Decentralization metrics for consensus }
\label{tab:metrics_symbols}
\begin{tabular}{|l|l|l|l|}
\hline
\textbf{Symbol} & \textbf{Metric}                                                          & \textbf{Range}                                                                            & \textbf{Ideal} \\ \hline \hline
\( m \)               & Validator set cardinality                                                & \( m > 0\)                                                                          & higher        \\ \hline
\( G \)               & Gini Index                                                               &  \( 0 \leq G \leq 1\)                                                                & lower          \\ \hline
\( \rho_{\mathbb{N}_L} \)             & \begin{tabular}[c]{@{}l@{}}Nakamoto Coefficient \\ for Liveness\end{tabular} & \begin{tabular}[c]{@{}l@{}}\( 0\leq \rho_{\mathbb{N}_L} \leq 100\)\end{tabular} & higher         \\ \hline
\( \rho_{\mathbb{N}_S} \)               & \begin{tabular}[c]{@{}l@{}}Nakamoto Coefficient\\ for Safety\end{tabular}    & \begin{tabular}[c]{@{}l@{}}\( 0\leq \rho_{\mathbb{N}_S} \leq 100\)\end{tabular} & higher        \\ \hline
HHI               & \begin{tabular}[c]{@{}l@{}}Herfindahl-Hirschman\\ Index\end{tabular}    & \begin{tabular}[c]{@{}l@{}}\( 0\leq \text{HHI} \leq 1\)\end{tabular} & lower        \\ \hline
\( G_{\varphi^{L}} \)       & \begin{tabular}[c]{@{}l@{}}Shapley for Liveness\\ Gini Index\end{tabular}    & \begin{tabular}[c]{@{}l@{}}\( 0\leq G_{\varphi^{L}} \leq 1\)\end{tabular} & lower        \\ \hline
\( G_{\varphi^{S}} \)       & \begin{tabular}[c]{@{}l@{}}Shapley for Safety\\ Gini Index\end{tabular}    & \begin{tabular}[c]{@{}l@{}}\( 0\leq G_{\varphi^{S}} \leq 1\)\end{tabular} & lower        \\ \hline
\( \mathcal{Z} \)       & \begin{tabular}[c]{@{}l@{}}Zipf's Law\\ Coefficient\end{tabular}    & \begin{tabular}[c]{@{}l@{}}\( 0\leq \mathcal{Z}\)\end{tabular} & lower        \\ \hline
\end{tabular}

\end{table}

\subsection{Nakamoto Coefficient - Liveness (\(\rho_{\mathbb{N}_L} \))}
Quantifies the minimum percentage of validators~\cite{balajidecentralization} required to disrupt block production~\cite{zhang2022reaching} or censor transactions~\cite{censorshipData}. \( \rho_{\mathbb{N}_L} \) encapsulates fault tolerance in weighted consensus. Specifically:

\begin{equation}
    \rho_{\mathbb{N}_L} = \frac{\mathbb{N}_L}{m} \times 100
\end{equation}

where \(\mathbb{N}_L \) is cardinality of the smallest subset of validators $L$ whose cumulative weight is at least one-third of the total:
\begin{equation}
    \mathbb{N}_L = \min\{|L| \mid L \subseteq N, \sum_{i \in L} w_i \geq \frac{1}{3} W\}
\end{equation}

This percentage facilitates comparison across blockchains of varying sizes. Unlike the \( (m,\varepsilon,\delta) \)-decentralization model, which focuses on the richest validator, \( \rho_{\mathbb{N}_L} \) considers the weight distribution across the top one-third of validators relative to the entire set, offering a broader view of the system's decentralization.

Inference: A higher \( \rho_{\mathbb{N}_L} \) (closer to 100\%) indicates better decentralization and resilience to liveness attacks~\cite{censorshipData}.

Limitations: Firstly, \( \rho_{\mathbb{N}_L} \) must be interpreted alongside the total validator count \( m \), similar to $G$. Secondly, \( \rho_{\mathbb{N}_L} \) assumes correct behavior among validators. Any collusion among validators may distort \( \rho_{\mathbb{N}_L} \), rendering it an inaccurate measure of decentralization.

\subsection{Nakamoto Coefficient - Safety (\( \rho_{\mathbb{N}_S} \))}
Safety relies on finality—ensuring once blocks are appended, they become immutable. Safety compromises have severe consequences on the ledger's integrity, such as ledger re-ordering or loss of funds~\cite{li2023hard,sridhar2023better}. The Nakamoto Coefficient for Safety (\(\rho_{\mathbb{N}_S} \)) is the minimum percentage of the validators required to compromise safety~\cite{balajidecentralization,pass2017analysis}. Specifically, \( \rho_{\mathbb{N}_S} \) is given as:
\begin{equation}
\rho_{\mathbb{N}_S} = \frac{\mathbb{N}_S}{m}
\end{equation}
where \( \mathbb{N}_S \) is the cardinality of the smallest subset of validator set $S$ whose combined weight can form a quorum \( \mathbb{Q} \):
\begin{equation}
\mathbb{N}_S = \min\{|S| \,|\, S \subseteq N, \sum_{i \in S} w_i \geq \mathbb{Q} \}
\end{equation}
In conjunction with \(\rho_{\mathbb{N}_L} \), \( \rho_{\mathbb{N}_S}\) complements the \( (m,\varepsilon,\delta) \)-decentralization framework by quantifying the concentration of weight that affects system safety.
    
Inference: Higher values of \( \rho_{\mathbb{N}_S} \), close to 100, especially in systems with a large validator set (\( m \)), signifies robust decentralization.

Limitations: The assumption that validators in \( \rho_{\mathbb{N}_S} \) calculations are non-colluding may not reflect real-world scenarios, potentially limiting its accuracy for safety evaluation. Moreover, reducing a system's safety to a single metric like \( \rho_{\mathbb{N}_S} \) risks oversimplifying safety complexities~\cite{wang2021ethereum}. Needless to say, the interpretation of \( \rho_{\mathbb{N}_S} \) is incomplete without considering \( m \).

\subsection{Herfindahl-Hirschman Index (HHI)}

The HHI is a measure commonly used by regulatory authorities to determine market concentration and competitiveness~\cite{calkins1983new}. In consensus mechanisms, HHI measures the concentration of stake among validators. Let \( w_k \) represent the weight of validator \( k \), and \( w_k' \) be its normalized weight. The HHI is defined as:

\begin{equation}
    HHI = \sum_{k=1}^{m} (w_k')^2, \text{ where } w_k' = \frac{w_k}{\sum_{k=1}^{m} w_k}
\end{equation}

The $HHI$ complements the \( (m, \epsilon, \delta) \)-decentralization model by quantifying stake concentration, thereby offering an additional perspective on weight distribution.

Inference: The HHI ranges from 0 (maximum decentralization) to 1 (complete centralization).

Limitations: The non-linear emphasis on large stake weights reduces sensitivity to moderate concentrations, potentially masking oligopolistic structures; i.e., if there are multiple powerful validators, $HHI$ fails to capture decentralization. Similar to $G$, \( \rho_{\mathbb{N}_L} \) and \( \rho_{\mathbb{N}_S} \), $HHI$ does not consider \( m \), i.e., a small validator set can exhibit a low $HHI$ yet still be centralized.

\subsection{Shapley Value for Liveness and Gini Analysis (\( G_{\varphi^{L}} \))}
The Shapley value, from cooperative game theory~\cite{winter2002shapley}, is used to quantify each validator's marginal contribution to consensus liveness. By evaluating all possible coalitions involving validator \( k \), the Shapley value for liveness (\( \varphi^{L} \)) of validator \( k \) is given by:

\begin{equation}
    \varphi^{L}_k = \sum_{S \subseteq N \setminus \{k\}} \frac{|S|! \, (m - |S| - 1)!}{m!} \left( v(S \cup \{k\}) - v(S) \right)
\end{equation}
where \( S \subseteq N \setminus \{k\} \) is any subset of validators excluding \( k \), and the value function \( v(S) \) is defined based on the 33\% liveness threshold in classical consensus, as follows:

\begin{equation}
    v(S) = \begin{cases}
        1, & \text{if } \sum_{i \in S} w_i > 0.33 \times W \\
        0, & \text{otherwise}
    \end{cases}
\end{equation}

Specifically, \(v(S) = 1\) indicates that the coalition has the power to halt the system or censor transactions. The marginal contribution of validator \( k \) depends on whether their inclusion allows the coalition to cross this liveness threshold.

Unlike previous metrics, \( \varphi^{L} \) captures the coalitional influence of validators, emphasizing their roles in coalition dynamics and maintaining liveness.

Inference: Each validator \( k \) has an individual Shapley value \( \varphi^{L}_k \). To derive a unified system-wide decentralization metric, we calculate the Gini coefficient over Shapley values, denoted by \( G_{\varphi^{L}} \), to capture inequality in coalitional influence. \( G_{\varphi^{L}} = 0 \) implies equal contributions from all validators, indicating a highly decentralized system.

Limitations: Calculating \( \varphi^{L} \) is computationally intensive due to the factorial growth of subsets, making it impractical for large networks. It is also sensitive to validators with significant weights and assumes influence is solely based on weights, ignoring network topology.

\subsection{Shapley Value for Safety and Gini Analysis (\( G_{\varphi^{S}} \))}

We extend our analysis of coalition dynamics to safety using the Shapley value for safety, denoted by \( \varphi^{S} \). Similar to \( \varphi^{L} \), it quantifies each validator's marginal contribution, with the difference being in the value function \( v(S) \), which now reflects the 66\% safety threshold in classical consensus:

\begin{equation}
    v(S) = \begin{cases}
        1, & \text{if } \sum_{i \in S} w_i > 0.66 \times W \\
        0, & \text{otherwise}
    \end{cases}
\end{equation}

Inference: Each validator \( k \) has an individual Shapley value \( \varphi^{S}_k \), and we derive a unified system-wide metric using the Gini coefficient, denoted by \( G_{\varphi^{S}} \), to capture inequality in safety contributions. A low Gini coefficient, i.e., \( G_{\varphi^{S}} \to 0 \), implies higher decentralization.

Limitations: The limitations for \( G_{\varphi^{S}} \) are identical to those described for \( G_{\varphi^{L}} \), including computational complexity and sensitivity to validator stake distribution.

\subsection{Zipf's Law Coefficient (\( \mathcal{Z} \))}

Zipf's Law states that when a list of values is sorted in decreasing order, each value is approximately inversely proportional to its rank raised to an exponent~\cite{zipf2016human}. This pattern is commonly found in linguistics for word frequencies~\cite{piantadosi2014zipf} and in economics for wealth distributions~\cite{newman2005power}, offering insights into inequality.

In consensus mechanism, the Zipf's law coefficient (\( \mathcal{Z} \)) analyses validator weight distribution across ranks. Specifically, if \( w_r \) represents the consensus weight of the validator ranked \( r \), the distribution follows:
\vspace{-4pt}
\begin{equation}
    w_r \propto r^{-\mathcal{Z}}
\end{equation}
where \( \mathcal{Z} \) is the Zipf exponent. \( \mathcal{Z} \) adds a unique perspective to the \( (m, \epsilon, \delta) \)-decentralization model. Unlike the Gini index or Shapley values, \( \mathcal{Z} \) provides a rank-based assessment of stake distribution and its adherence to a power-law pattern.

Inference: A higher \( \mathcal{Z} > 2 \) suggests that a few validators hold most of the consensus weight, indicating centralization. Conversely, \( \mathcal{Z} \to 0 \) reflects a more uniform distribution, indicating better decentralization.

Limitations: \( \mathcal{Z} \) is an empirical observation with limited theoretical basis. It assumes a power-law relationship and is sensitive to rank ordering, making it vulnerable to statistical noise and unable to account for external incentives that might alter distribution patterns.

\subsection{Summary}
In this section, we introduced various metrics to quantify decentralization in consensus mechanisms, see Table~\ref{tab:metrics_symbols}. Despite the limitations acknowledged, the synergy of these metrics holistically captures the essence of decentralization, aligning with the \( (m,\varepsilon,\delta) \)-decentralization model. We leverage these metrics to quantify the decentralization of existing blockchain systems in practice in the subsequent section.

\section{Empirical Data Analysis}
\label{sec:empiricalanalysis}
\begin{table*}[htp]
\vspace{-12pt}
\caption{Decentralization metrics for permissionless blockchains.}
\centering
\label{tab:metrics}
\renewcommand{\arraystretch}{1.5}
\begin{tabular}{|l|l|l|l|l|l|l|l|l|l|l|l|}
\hline
\multicolumn{1}{|c|}{} & \textbf{Application} & \textbf{\( m \)} & \( G \) & \( \rho_{\mathbb{N}_L} \) & \( \rho_{\mathbb{N}_S} \) & HHI &\( G_{\varphi^{L}} \)  &  \( G_{\varphi^{S}} \) & \( \mathcal{Z} \) & \textbf{\( \varepsilon \) in \( (m,\varepsilon,0) \)} & \( (m,\varepsilon,50) \) \\ \hline \hline
\textbf{Aptos}         & L1 blockchain~\cite{aptos}   & 191        & 0.55       & 11.52 & 27.23 & 0.011 &0.561   & 0.559 & 3.135  & 1127805207331 &  7.99   \\ \hline
\textbf{Axelar}        & Interoperability~\cite{axelar}         & 75         & 0.37       & 16&41.33 & 0.02 &0.375        & 0.373 &0.928 & 591.9 & 2.83 \\ \hline
\textbf{BNB (Binance)} & L1 blockchain~\cite{bnb}           & 57         & 0.55       & 14.04 &28.07& 0.036 &0.556         & 0.558 &2.316  & 159511.43 & 8.41 \\ \hline
\textbf{Celestia} & Data availability~\cite{celestia}          & 239         & 0.8       & 2.93&11.3&0.026 &0.797         & 0.803 & 5.679 & 4468193447.12 & 4940.77 \\ \hline
\textbf{Celo}  & L2* (L1 blockchain)~\cite{celo}                     & 86         & 0.42&13.95&38.37       & 0.019 &0.426       & 0.425 &1.517 & 8805762239.16 & 2.85       \\ \hline
\textbf{Cosmos}        & L1/interoperability~\cite{cosmos}        & 200 & 0.72       & 3.5 & 12  & 0.029 & 0.73&0.731 & 1.573 & 23841.39 & 72.73\\ \hline
\textbf{Injective} &  DeFi/interoperability~\cite{injective}        & 60        & 0.4       & 11.67 & 38.33          & 0.033 &0.409&0.414 &0.72   & 26.62 & 8.96      \\ \hline
\textbf{Osmosis}       & DeFi/DEX~\cite{osmosis}       & 150        & 0.54       & 6.67&28&0.02 &0.554         & 0.557 &1.022    & 108.05 & 14.52            \\ \hline
\textbf{Polygon}       & L2/ZK-rollup~\cite{polygon}                  & 105        & 0.79       & 3.81 &11.43         & 0.051 & 0.8 & 0.8 & 2.057   & 35677.03 & 114.84         \\ \hline
\textbf{Sui}           & L1 blockchain~\cite{sui}     & 108        & 0.35       & 15.74 & 39.81        & 0.014 & 0.351 &0.358 &0.63 & 10.37 & 3.0     \\ \hline
\end{tabular}
\end{table*}
\subsection{Scope and Methodology of Data Collection}
In our study, we focus on permissionless blockchains that use classical consensus mechanisms, particularly those with weighted consensus as outlined in Section~\ref{sec:consensusclassification}. We limit our analysis to blockchains with deterministic validator set selection methods, such as PoS and DPoS. Notably, all the blockchains examined in this study employ DPoS for Sybil resistance.

For our empirical analysis, we selected ten blockchains, as shown in Table~\ref{tab:metrics}: Aptos, Axelar, BNB (Binance), Celestia, Celo, Cosmos, Injective, Osmosis, Polygon, and Sui. Most of the selected blockchains use CometBFT (formerly Tendermint~\cite{buchman2016tendermint}) as their consensus mechanism~\cite{injectiveBlog, bnbBlockchain, celestiaDocs, axelarBlockchain, cosmosBlockchain, polygonBlockchain}, while Aptos, Sui, and Celo use Jolteon~\cite{gelashvili2022jolteon}, Mysticeti~\cite{babel2023mysticeti}, and IstanbulBFT~\cite{celoBlockchain, moniz2020istanbul}, respectively.

These blockchains were chosen for their diverse applications, including smart contracts Layer-1 (L1) blockchains, interoperability protocols~\cite{belchior2021survey}, Layer-2 (L2) scaling solutions~\cite{mccorry2021sok}, data availability protocols~\cite{al2019lazyledger}, and decentralized exchanges~\cite{lehar2021decentralized}, ensuring the wide applicability of our findings. Our sample comprises blockchains with at least 300 million USD market capitalization, cumulatively amounting to a market capitalization of 60 billion USD as of December 14, 2023~\cite{coinmarketcap}.

Data collection was automated via scripts interfacing with the blockchains' RPC endpoints to fetch active validator sets and their staked tokens. Daily snapshots were taken to account for epoch changes, and the data was systematically archived in a public GitHub repository\footnote{https://github.com/sm86/destake} to facilitate transparency and accessibility for ongoing and subsequent analyses.

\subsection{Data Analysis}
In Table~\ref{tab:metrics}, we present decentralization metrics for ten blockchains based on validator set data \textcolor{black}{using snapshots on} October 25, 2024, except for Celestia and Binance, where data \textcolor{black}{snapshots on December 14, 2023 were} used due to deprecated RPC endpoints. To ensure reliability, we monitored the data over several months and confirmed stable trends.

Validator set cardinality (\( m \)) shows a range of 57 to 239 validators across analyzed blockchains. Most validator sets remained constant or saw minor increases since December 2023, while Celestia and Aptos increased \( m \) significantly by 33\%. Despite these validator set cardinalities, Nakamoto coefficients reveal vulnerabilities: only 2.93\% to 16\% of validators are required to compromise liveness (\( \rho_{\mathbb{N}_L} \)), and 11.3\% to 41.33\% for safety (\( \rho_{\mathbb{N}_S} \)). For instance, on Polygon, the top four validators could censor an application, and a coalition of the top twelve validators could alter the ledger. These findings underscore the need to strengthen blockchain liveness and safety using available validator sets.

An examination of the Gini coefficient (\( G \)) in Table~\ref{tab:metrics}, ranging from 0.35 to 0.8, reveals a significant concentration of weight, indicating disproportionate stake distribution among a small subset of validators, particularly in Celestia, Polygon, and Cosmos. The implications of stake concentration are further underscored by the values of $\varepsilon$ in the $(m, \varepsilon, \delta)$-decentralization model. Both at $\delta = 0$ and $\delta = 50$, the $\varepsilon$ values substantially deviate from the ideal zero value needed for $(m, 0, 0)$, i.e., full decentralization. While most blockchains demonstrate modest but concerning trends in Gini and Nakamoto coefficients since December 2023, Axelar and Sui show a promising trend. 

We observe low HHI values across all blockchains ($\textsc{max}(\text{HHI}) \leq 0.05$), indicating the absence of a single dominant validator, aligning with decentralization ideals. However, this highlights a limitation of HHI: while high HHI values confirm monopolistic centralization, low HHI values alone do not fully capture decentralization, as they mask an uneven distribution of influence among an oligopoly of validators.

The Gini indices for the Shapley values, both for liveness (\( G_{\phi_L} \)) and safety (\( G_{\phi_S} \)), are highly similar across all blockchains. This similarity indicates that validators' contributions to both liveness and safety are nearly identical, which reflects the skewed stake distribution. More notably, the Gini indices for Shapley values closely match the Gini index of stake weights (\( G \)). Further analysis reveals that the Pearson correlation between stake weights and Shapley values is close to one (\( \sim 0.9988 \)). This suggests that, despite the cooperative game-theoretic foundation of Shapley values—intended to quantify the marginal contribution of each validator in the consensus—their similarity with stake-weights underscores that influence is predominantly determined by stake concentration. While Shapley values might seem redundant, they validate the robustness of other metrics and emphasize the influence from economic stake.

Finally, we analyze Zipf's coefficient (\( \mathcal{Z} \)). The values of \( \mathcal{Z} \) vary significantly among the blockchains. Celestia, Aptos, Binance, and Polygon exhibit high \( \mathcal{Z} \), indicating concentrated stake distributions that hinder decentralization. In contrast, Sui, Injective, and Axelar have lower, more ideal \( \mathcal{Z} \) values. This difference likely arises from mechanisms such as minimum validator thresholds (e.g., Sui) or fixed validator sets (e.g., Axelar), which limit eligibility and diverge from the power law, resulting in lower \( \mathcal{Z} \).

\subsection{Inferences and Open Challenges}
We observe that while validator sets are adequately sized, influence within these sets remains disproportionately concentrated. This analysis raises key questions: for a validator set of cardinality \( m \), how can we enhance Nakamoto coefficients? Furthermore, how can we achieve a more equitable weight distribution across metrics, specifically \( G \), HHI, and \( \mathcal{Z} \)?

A straightforward solution could be a one-validator, one-vote system with equal weight for all validators and a minimum stake threshold for eligibility. However, blockchains like Ethereum illustrate practical limitations. With a low threshold of 32 ETH, Ethereum’s network includes over 800,000 validators~\cite{ethereumCoindesk}, necessitating random selection for consensus, which impacts performance metrics like finality time. A larger issue is operational centralization, as entities like Lido control approximately 32.7\% of all validators~\cite{lidoArticle}, resulting in a low Nakamoto liveness coefficient (\(\mathbb{N}_L\approx1\)).

This highlights the challenge of preventing Sybil attacks while promoting genuine decentralization—specifically, reducing incentives to create multiple Sybil identities. Although modeling Sybil resistance is complex, we address pressing challenges of equitable stake distribution and improving Nakamoto coefficients for a given validator set in the following section.

\section{Advancing Decentralization: SRSW and LSW Quorums}
\label{sec:srsw-model}
In this section, we address the identified challenges. We begin by defining primitives needed for our proposed solution. This is followed by the two variants of updated consensus weight computations: SRSW (Square Root Stake Weighted) and LSW (Logarithmic Stake Weighted) quorums.

\subsection{Primitives and Assumptions}

\textbf{Validator Rewards:} 
Validators are incentivized with rewards for their participation in the consensus mechanism. We assume validators are rational and want to maximize their rewards. 

At the end of every epoch, each correct validator \( n_k \in N \) receives a reward, denoted as \( r_{n_k} \). The reward is calculated based on the system parameter \( \alpha \), representing the inflation factor that determines the rate of reward distribution. The reward for each validator is given by the equation:
\vspace{-4pt}
\begin{equation}
    r_{n_k} = \alpha w_k, \quad \text{where} \quad \alpha > 0
\end{equation}

If a correct validator holds \( s_k \) native tokens at the start of an epoch, their balance of native tokens after that epoch would increase to \( s_k + r_{n_k} \).

\textbf{Sybil Cost:}
We introduce a Sybil Cost, denoted as \( C > 0 \), to represent the additional expenses incurred by a validator operator when choosing to run multiple validator nodes instead of one. These expenses could include operational costs, such as computational resources or the amount of staked tokens required. We assume that \( C \) is sufficiently high, thereby providing resistance against Sybil attacks.

\textbf{Limit Validator Set Cardinality:}
We propose an upper limit for the validator set cardinality  \( M \), ensuring \( m \leq M \). To implement this in practice, the validator candidates are sorted based on their staked tokens, and we select the top \( M \) candidates for the validator set, i.e., the threshold staked tokens to become a validator is the stake of the \( M \)th validator candidate, represented by \( s_{M} \). Accordingly, the rewards for a validator candidate \( n_k \) with \( s_k \) staked tokens for an epoch is as follows: 
\begin{equation}
r_{n_k}=\begin{cases}
          \alpha w_k \quad &\text{if} \, s_k > s_M,  \\
          0 \quad &\text{if} \, s_k \leq s_M. \\
     \end{cases}
\end{equation}

Capping the validator set cardinality is justified for two reasons. Firstly, in line with classical consensus mechanisms, an increase in the number of validators tends to decrease the system scalability, measured in throughput and latency~\cite{yin2019hotstuff}. Secondly, by imposing a maximum limit on the number of validators, and a minimum capital requirement of \( s_{M} \) staked tokens for running a validator, the system discourages single entities from dominating the validator set (i.e., Sybil attacks), further explored in the following section. This design, implemented in blockchains such as Axelar~\cite{axelarBlockchain} and Celo~\cite{celoBlockchain}, helps in balancing scalability and decentralization.

\vspace{-8pt}
\subsection{SRSW Function}
Building on these primitives, we now focus on the challenge of achieving equitable influence among validators to improve the metrics such as Nakamoto and Gini coefficients for a given validator set.

We propose the \textit{Square Root Stake Weight (SRSW)} function, a novel approach that diverges from traditional linear weightings in quorum \( \mathbb{Q} \) computations. The SRSW function calculates the weight \( w^*_i \) of each validator \( n_i \) based on the square root of their staked tokens \( s_i \), as defined by:
\vspace{-4pt}
\begin{equation}
   w^*_i = \sqrt{s_i} 
   \vspace{-4pt}
\end{equation}

The revised quorum \( \mathbb{Q}^* \) for the validator set \( N \) is formulated as follows:
\vspace{-6pt}
\begin{equation}
\mathbb{Q}^* \geq \left( \frac{2}{3}\right) \sum_{i} w_i^* = \left( \frac{2}{3}\right) \sum_{i} \sqrt{s_i} \quad \forall i \in N
  \vspace{-4pt}
\end{equation}

Contrasting with linear models, the SRSW function aims to reduce the disproportionate influence of validators with high staked tokens. In essence, the SRSW approach diminishes the weight disparities between validators with varying stake amounts. 

The validator rewards \( r^*_{n_i} \) are structured to reinforce rational decisions. The reward formula is:
\begin{equation}
    r^*_{n_i} = \alpha w_i^* = \alpha \sqrt{s_i}, \text{ if } s_i > s_M; \text{ otherwise, } 0.
\end{equation}
This incentivizes validators to keep or increase their stakes above the threshold \( s_M \), aligning individual gains with the system's stability. In other words, the system should satisfy the following condition. 
\vspace{-4pt}
\begin{equation}
    r_{n_i}^* > r_{n_j}^* + r_{n_k}^* - C, \text{ where } s_i \geq s_j + s_k
    \vspace{-4pt}
\end{equation}
This inequality implies that a validator with a combined stake \( s_i \) gains more rewards by maintaining a single identity rather than dividing into multiple validators with smaller stakes \( s_j \) and \( s_k \).

Consider a validator with \( s_i = 4 \) and \( s_M = 3 \). The options are: split into \( s_j^1 = 2, s_k^1 = 2 \), or \( s_j^2 = 3, s_k^2 = 1 \), or not split. In the first case, \( s_j^1, s_k^1 < s_M \) yield no rewards. In the second, rewards are \( \alpha \sqrt{3} \) for \( s_j^2 \) only, as \( s_k^2 < s_M \). Not splitting, \( \alpha \sqrt{4} \), offers the highest reward. Our approach effectively deters stake fragmentation and mitigates Sybil attacks, promoting consolidated stakes as a strategically rational choice.

When both \( s_i \) and \( s_j \) exceed \( s_M \), the inequality is adjusted in terms of weights:
\vspace{-2pt}
\begin{equation}
    \sqrt{s_i} > \sqrt{s_j} + \sqrt{s_k} - \frac{C}{\alpha},\text{ where } s_i \geq s_j + s_k \text{ and } s_i, s_j > s_M
\end{equation}

Here, our assumption of high \( C \) plays a crucial role to make the division of validator stakes non-rational, which we explore in the subsequent sections.

\subsection{LSW Function}
We propose the Log Stake Weight (LSW) function to improve consensus decentralization. Similar to SRSW, LSW diverges from traditional linear weighting but, unlike SRSW, adopts a logarithmic function. Specifically, LSW calculates the weight \( w^\phi_i \) of each validator \( n_i \) based on the logarithm of their staked tokens \( s_i \), given by:

\vspace{-4pt}
\begin{equation}
    w^\phi_i = \log{s_i}
\vspace{-4pt}
\end{equation}

The quorum \( \mathbb{Q}^\phi \) for the validator set \( N \) is defined as:

\vspace{-6pt}
\begin{equation}
\mathbb{Q}^\phi \geq \left( \frac{2}{3}\right) \sum_{i} w_i^\phi = \left( \frac{2}{3}\right) \sum_{i} \log{s_i} \quad \forall i \in N
\vspace{-4pt}
\end{equation}

This logarithmic transformation reduces the impact of validators with disproportionately large stakes without negating their influence, promoting a balance between decentralization and meaningful incentives for larger stakeholders.

Validator rewards in the LSW model are defined analogously to SRSW, align rewards with participation to incentivize consensus contribution:
\vspace{-8pt}
\begin{equation}
    r^\phi_{n_i} = \alpha w^\phi_i = \alpha \log{s_i}, \quad \text{if } s_i > s_M; \text{ otherwise, } 0.
\vspace{-4pt}
\end{equation}

This model incentivizes validators to maintain their stakes above a threshold \( s_M \) while the logarithmic weighting discourages stake fragmentation aimed at gaining disproportionate influence. As a result, LSW contributes towards improving decentralization.

\subsection{Discussion}
\textbf{Implementation Specifics.}
\textcolor{black}{Any PoS/DPoS system can transform to SRSW and LSW models quite easily as the only change would be on the voting power computations. This could happen at reconfiguration events at epoch boundaries where validator set $V$ is updated. Every validator in the protocol can compute the choose approach, whether it is LSW or SRSW, to compute the new voting power.}

\textcolor{black}{Transitioning existing PoS/DPoS systems to SRSW or LSW models is feasible by modifying voting power computations during reconfiguration events at epoch boundaries when the validator set \(V\) is updated. At each epoch boundary \(t\), all validators \(v_i \in V\) compute their new voting power \(w_i(t)\) using either the SRSW or LSW function:}

\begin{equation}
    w_i(t) = 
    \begin{cases}
        \sqrt{s_i(t)} & \text{for SRSW} \\
        \ln(1 + s_i(t)) & \text{for LSW}
    \end{cases}
\end{equation}

\textcolor{black}{where \(s_i(t)\) is the stake of validator \(v_i\) at time \(t\). These weights are then used in subsequent consensus rounds for voting and block proposal selection. This approach allows a smooth transition without requiring fundamental changes to the underlying consensus mechanism.}

\textbf{Determining M.}
A critical aspect of this approach is the determination of \( M \), the maximum cardinality of the validator set. A low \( M \) might risk insufficient decentralization, while an excessively high \( M \) could impact system performance due to increased communication complexity. Although a fixed \( M \) may appear counterintuitive to decentralization, delegation in DPoS mechanisms enable individual token holders to collectively participate in consensus, thereby mitigating potential centralization concerns~\cite{saad2020comparative}. In this work, we do not prescribe a specific value for \( M \), as it depends on the algorithm and implementation. In practice, we observed around one hundred validators is the ideal number for current algorithms~\cite{yin2019hotstuff,0LNetwork}.

\textbf{Sybil costs.} 
\textcolor{black}{
To address Sybil attack concerns in non-linear staking models, we propose concrete implementations of Sybil costs (\(C\)). For SRSW rewards (\(R(S)=\sqrt{S}\)), the minimum \(C\) needed to prevent stake splitting is derived as follows: Consider a validator with stake \(S\) splitting into \(n\) identities. The reward for a single identity is \(\sqrt{S}\), while for \(n\) split identities it is \(n \cdot \sqrt{S/n}\). To avoid splitting, we require \(\sqrt{S} > n \cdot \sqrt{S/n} - C\), yielding \(C \geq (\sqrt{n} - 1)/(n-1) \cdot \sqrt{S}\). }

\textcolor{black}{For LSW, rewards (\(R(S)=\ln(1+S)\)), a similar inequality applies: The reward for a single identity is \(\ln(1+S)\), while for \(n\) split identities it is \(n \cdot \ln(1+S/n)\). To avoid splitting, we require \(\ln(1+S) > n \cdot \ln(1+S/n) - C\), approximating to \(C \geq \ln(n)/(n-1)\) for large \(S\).}

\textcolor{black}{The minimum \(C\) can be implemented through various mechanisms: (1) stake-weighted resource limits similar to Solana's Quality of Service mechanism~\cite{solanaQos}, (2) minimum stake thresholds comparable to Ethereum's 32 ETH requirement~\cite{ethereumCoindesk}, (3) reputation-based penalties for new identities, (4) geospatial validation mechanisms using spatial coordinates to limit one validator per physical location~\cite{motepalli2023analyzing}, and (5) Proof of Personhood~\cite{worldcoin,borge2017proof,hodgson2002know} to verify unique human identity without compromising anonymity. These mechanisms provide necessary economic disincentives against Sybil attacks while preserving the decentralization advantages of non-linear stake weighting.}

Furthermore, detecting cartels among validators is challenging, especially at the protocol layer; therefore, we do not address this issue in the consensus mechanism. In conclusion, we acknowledge that establishing Sybil costs is more a complex socio-economic challenge than a technical one and is beyond the scope of this work.

In summary, we propose the SRSW and LSW function with updated quorums and reward functions, and provide considerations on \( M \) and \( C \). We now turn our attention to evaluating these functions.

\section{Evaluation: Improved Decentralization}
\label{sec:evaluation}

In this section, we evaluate the decentralization effectiveness of the SRSW and LSW functions compared to the linear model, using the metrics defined in Section~\ref{sec:metrics}. We begin with a formal theoretical analysis, followed by an empirical evaluation.

\begin{figure*}[htbp]
  \centering
  \includegraphics[trim={0.6cm 0.6cm 0.6cm 1cm}, width=\textwidth]{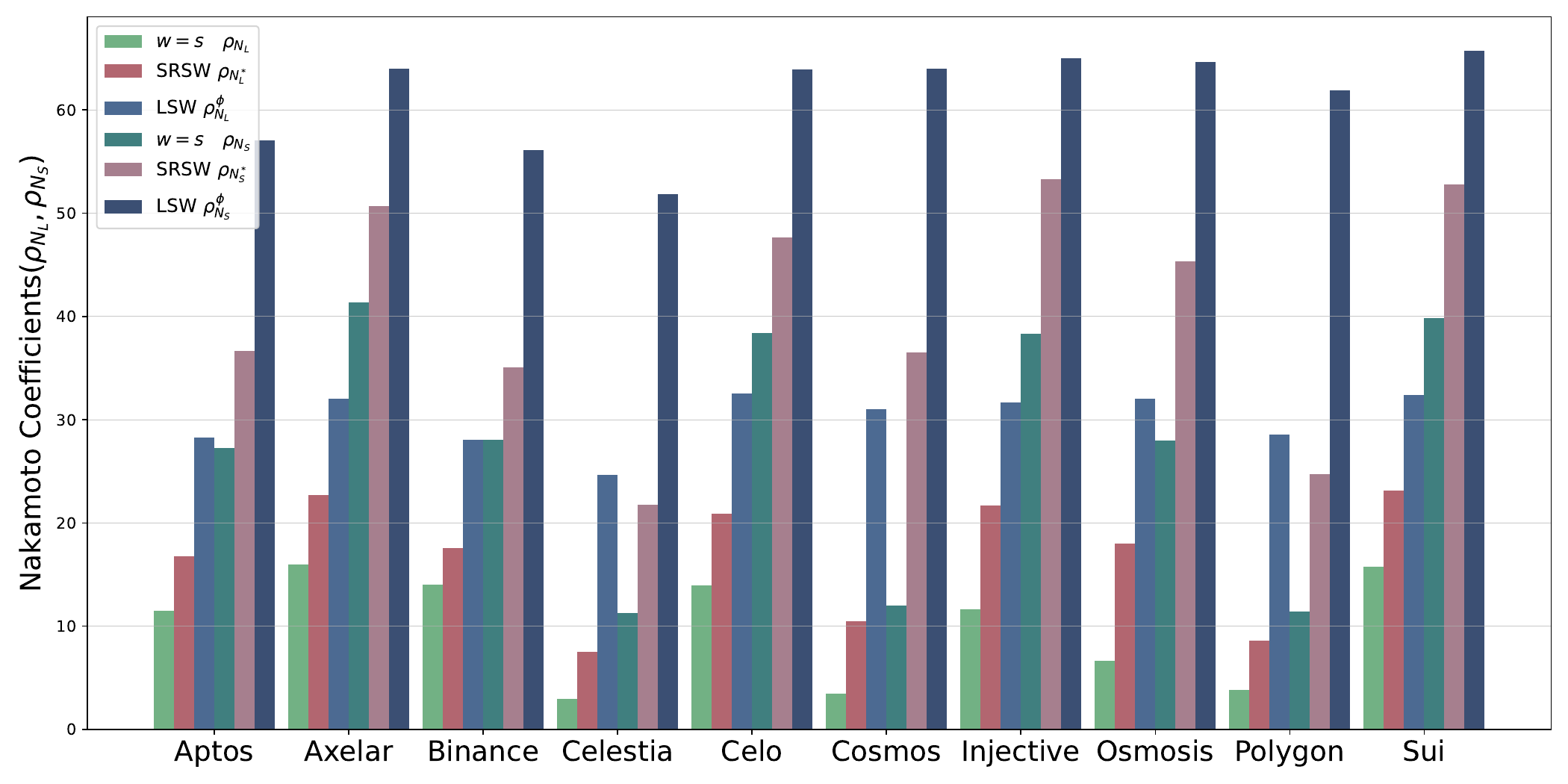}
  \caption{Comparison of Nakamoto coefficients for safety and liveness in linear, SRSW and LSW weighting functions}
  \label{fig:nakamoto_coefficients}
  
\end{figure*}

\subsection{Formal Analysis of Decentralization Metrics}
We demonstrate how the SRSW and LSW models improve decentralization over the linear stake-weight model using established metrics.

\vspace{-4pt}
\begin{lemma}
\label{lemma:nakamoto-srsw}
Given a validator set \( N \), the SRSW model's Nakamoto coefficients, for liveness \( \rho_{\mathbb{N}_L}^* \) and safety \( \rho_{\mathbb{N}_S}^* \), are greater than or equal to that of the linear stake-weight model, represented as \( \rho_{\mathbb{N}_L} \) and \( \rho_{\mathbb{N}_S} \), respectively.
\end{lemma}

\begin{proof}
Consider the Nakamoto coefficient for liveness computation in linear stake weight, let $K$ be the smallest subset such that: 
\vspace{-2pt}
\begin{equation}
    \sum_{n_i\in K} s_i > \frac{1}{3} \sum_{n_i\in N} s_i
\end{equation}
Similary in the SRSW model, let $K^*$ be the smallest subset satisfying the condition:
\vspace{-4pt}
\begin{equation}
    \sum_{i \in K^*} \sqrt{s_i} > \frac{1}{3} \sum_{i\in N} \sqrt{s_i}
\end{equation}
The concave nature of the square root function, as per Jensen's inequality~\cite{abramovich2004refining}, necessitates a larger $K^*$ to fulfill this condition in the SRSW model compared to $K$ in the linear model, thereby implying:
\vspace{-6pt}
\begin{equation}
    |K^*| \geq |K| \implies N_L^* \geq N_L
\end{equation}
\begin{equation}
    \frac{N_L^*}{m} \geq  \frac{N_L}{m} \implies \rho_{\mathbb{N}_L}^* \geq \rho_{\mathbb{N}_L}
\end{equation}
Hence, the Nakamoto coefficient for liveness is higher for SRSW compared to linear model. Similarly, we can prove for Nakamoto coefficient-safety, i.e., $\rho_{\mathbb{N}_S}^* \geq \rho_{\mathbb{N}_S}$. 
\end{proof}

\begin{lemma}
\label{lemma:srsw-gini}
Given \( N \), the Gini of SRSW and linear models, represented by \( G^* \) and \( G \), respectively, satisfy \( G^* \leq G \).
\end{lemma}

\begin{proof}
This proof draws upon the established principles from Lemma~\ref{lemma:nakamoto-srsw} and the definitions provided in Section~\ref{sec:metrics}.

In the SRSW model, the square root transformation applied to validator stakes results in a more uniform distribution of weights, effectively reducing relative disparities in stake sizes compared to the linear model. Consequently, this leads to a lower Gini index in the SRSW model, compared to the Gini index in the linear model.

Therefore, \( G^* \leq G \), indicating a more equitable distribution of validator influence under the SRSW model.
\end{proof}
\vspace{-5pt}

\begin{theorem}
\label{theorem:srsw-good}
Given a validator set \( N \), the SRSW model achieves a degree of decentralization, across key metrics, that is greater than or equal to that of the linear stake-weight model.
\end{theorem}

\begin{proof}
Building on Lemma~\ref{lemma:nakamoto-srsw} and Lemma~\ref{lemma:srsw-gini}, the concave nature of the SRSW model reduces influence concentration by limiting large validators' dominance, resulting in lower HHI and Gini values. Additionally, due to its correlation with weights, the Shapley value distribution enhances decentralization. The Zipf’s coefficient also decreases, reflecting a balanced distribution among validators. Thus, across all metrics, the SRSW model achieves higher degree of decentralization compared to the linear model.
\end{proof}

\begin{theorem}
\label{theorem:lsw-superior}
Given a validator set \( N \), the LSW model achieves a higher degree of decentralization across key metrics than the SRSW model.
\end{theorem}

\begin{proof}
Let \( K^* \) be the smallest subset of validators in SRSW satisfying:
\begin{equation}
    \sum_{n_i \in K^*} \sqrt{s_i} > \frac{1}{3} \sum_{n_i \in N} \sqrt{s_i}.
\end{equation}
For LSW, let \( K^{\phi} \) be the smallest subset satisfying:
\begin{equation}
    \sum_{n_i \in K^{\phi}} \log(s_i) > \frac{1}{3} \sum_{n_i \in N} \log(s_i).
\end{equation}
Since \( \log(s) \) is more concave than \( \sqrt{s} \), Jensen’s inequality~\cite{abramovich2004refining} implies that a larger subset \( K^{\phi} \) is required for LSW to meet this threshold, yielding:
\begin{equation}
    |K^{\phi}| \geq |K^*| \implies \rho_{\mathbb{N}_L}^{\phi} \geq \rho_{\mathbb{N}_L}^*.
\end{equation}

This advantage extends to the Nakamoto coefficient for safety, as well as other decentralization metrics, including Gini, HHI, Shapley values, and Zipf's coefficient. Thus, the LSW model achieves a higher degree of decentralization than the SRSW model.
\end{proof}

\begin{corollary}
\label{corollary:lsw-linear}
By Theorem~\ref{theorem:srsw-good} and Theorem~\ref{theorem:lsw-superior}, and using transitivity, the LSW model achieves a higher degree of decentralization than the linear stake-weight model across all metrics.
\end{corollary}
\begin{figure}[ht]
    \includegraphics[width=\linewidth, trim= {0.6cm 0cm 1cm 1.2cm}, clip]{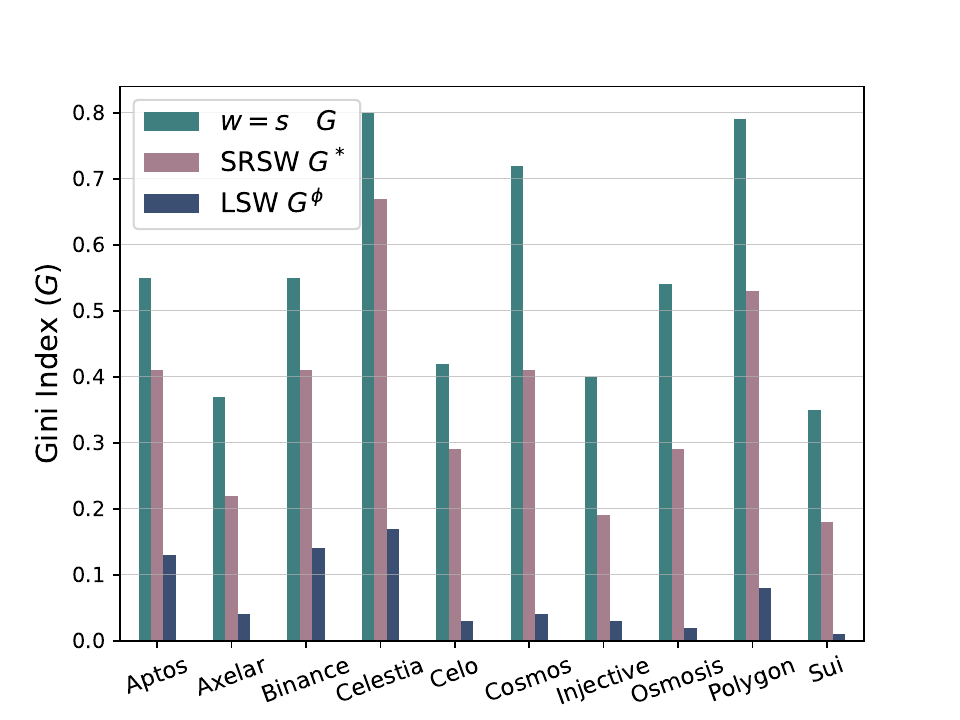}
        \caption{Comparison of Gini index}
        \label{fig:gini}
        \vspace{-12pt}
\end{figure}
\begin{figure}[ht]
   \includegraphics[width=\linewidth, trim= {1cm 0.8cm 1cm 1cm}, clip]{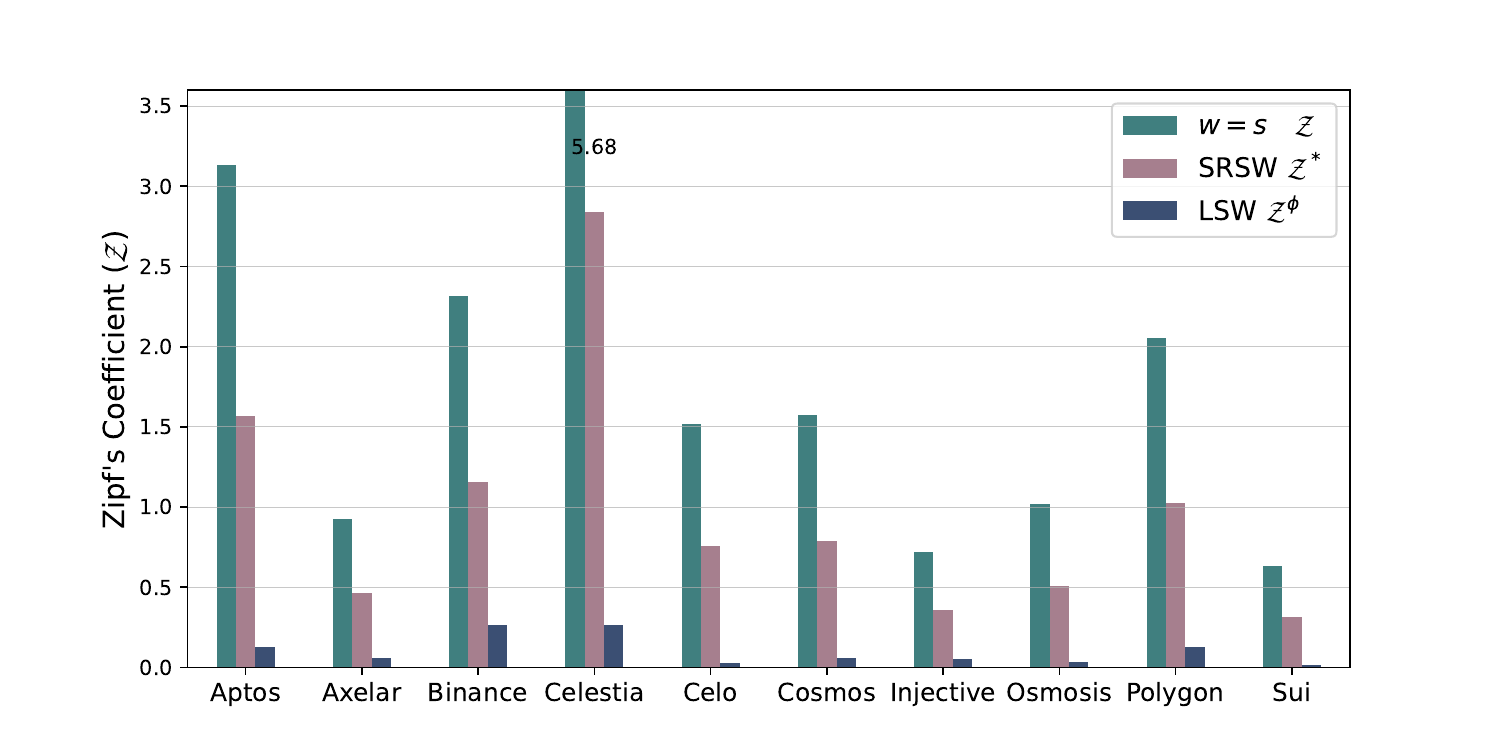}
        \caption{Comparison of Zipf's law coefficient}
        \label{fig:zipfs}
    
\end{figure}

\begin{table*}[ht]
\centering
\caption{Percentage improvements of \textcolor{blue}{SRSW} and \textcolor{red}{LSW} over linear token weights for various decentralization metrics}
\renewcommand{\arraystretch}{1.5} 
\setlength{\tabcolsep}{8pt} 

\begin{tabular}{|l|c|c|c|c|c|c|c|}
\hline
\label{tab:percentimprovement}
\textbf{Blockchain} & \textbf{\(\rho_{\mathbb{N}_L} \) (\%)} & \textbf{\(\rho_{\mathbb{N}_S} \) (\%)} & \textbf{$G$ (\%)} & \textbf{$HHI$ (\%)} & \textbf{\( G_{\varphi^{L}} \) (\%)} & \textbf{\( G_{\varphi^{L}} \)(\%)} & \textbf{$\mathcal{Z}$ (\%)} \\ \hline \hline
Aptos      & \textcolor{blue}{45.4} / \textcolor{red}{145.4} & \textcolor{blue}{34.59} / \textcolor{red}{109.59}  & \textcolor{blue}{25.45} / \textcolor{red}{76.36}  & \textcolor{blue}{27.27} / \textcolor{red}{45.45}  & \textcolor{blue}{26.56} / \textcolor{red}{72.19} & \textcolor{blue}{25.76} / \textcolor{red}{70.84} & \textcolor{blue}{49.98} / \textcolor{red}{95.92} \\ \hline
Axelar     & \textcolor{blue}{41.69} / \textcolor{red}{100}   & \textcolor{blue}{22.6} / \textcolor{red}{54.85}   & \textcolor{blue}{40.54} / \textcolor{red}{89.19}  & \textcolor{blue}{25} / \textcolor{red}{35}        & \textcolor{blue}{38.93} / \textcolor{red}{84.53}    & \textcolor{blue}{40.75} / \textcolor{red}{87.94} & \textcolor{blue}{50} / \textcolor{red}{93.32} \\ \hline
Binance    & \textcolor{blue}{24.93} / \textcolor{red}{99.93} & \textcolor{blue}{25.01} / \textcolor{red}{100}    & \textcolor{blue}{25.45} / \textcolor{red}{74.55}  & \textcolor{blue}{25} / \textcolor{red}{47.22}  & \textcolor{blue}{24.64} / \textcolor{red}{73.92}    & \textcolor{blue}{25.45} / \textcolor{red}{73.84} & \textcolor{blue}{50} / \textcolor{red}{88.51} \\ \hline
Celestia   & \textcolor{blue}{157} / \textcolor{red}{742.66}  & \textcolor{blue}{92.57} / \textcolor{red}{359.12} & \textcolor{blue}{16.25} / \textcolor{red}{78.75}  & \textcolor{blue}{53.85} / \textcolor{red}{80.77}  & \textcolor{blue}{15.56} / \textcolor{red}{76.54} & \textcolor{blue}{16.44} / \textcolor{red}{77.33} & \textcolor{blue}{49.99} / \textcolor{red}{95.33} \\ \hline
Celo       & \textcolor{blue}{50.04} / \textcolor{red}{133.41} & \textcolor{blue}{24.24} / \textcolor{red}{66.67}  & \textcolor{blue}{30.95} / \textcolor{red}{92.86} & \textcolor{blue}{21.05} / \textcolor{red}{36.84}  & \textcolor{blue}{30.05} / \textcolor{red}{83.1} & \textcolor{blue}{31.29} / \textcolor{red}{86.59} & \textcolor{blue}{50.03} / \textcolor{red}{98.09} \\ \hline
Cosmos     & \textcolor{blue}{200} / \textcolor{red}{785.71}  & \textcolor{blue}{204.17} / \textcolor{red}{433.33} & \textcolor{blue}{43.06} / \textcolor{red}{94.44}  & \textcolor{blue}{68.97} / \textcolor{red}{82.76}  & \textcolor{blue}{42.05} / \textcolor{red}{91.37} & \textcolor{blue}{42.41} / \textcolor{red}{90.15} & \textcolor{blue}{50.03} / \textcolor{red}{96.25} \\ \hline
Injective  & \textcolor{blue}{85.69} / \textcolor{red}{171.38} & \textcolor{blue}{39.13} / \textcolor{red}{69.58}  & \textcolor{blue}{52.5} / \textcolor{red}{92.5}   & \textcolor{blue}{42.42} / \textcolor{red}{48.48}  & \textcolor{blue}{53.3} / \textcolor{red}{89.73}    & \textcolor{blue}{50.48} / \textcolor{red}{92.51} & \textcolor{blue}{50} / \textcolor{red}{92.78} \\ \hline
Osmosis    & \textcolor{blue}{169.87} / \textcolor{red}{379.76} & \textcolor{blue}{61.89} / \textcolor{red}{130.96} & \textcolor{blue}{46.3} / \textcolor{red}{96.3} & \textcolor{blue}{55} / \textcolor{red}{65}        & \textcolor{blue}{47.29} / \textcolor{red}{91.7}    & \textcolor{blue}{45.52} / \textcolor{red}{87.2} & \textcolor{blue}{50} / \textcolor{red}{96.48} \\ \hline
Polygon    & \textcolor{blue}{124.93} / \textcolor{red}{649.87} & \textcolor{blue}{116.62} / \textcolor{red}{441.56} & \textcolor{blue}{32.91} / \textcolor{red}{89.87} & \textcolor{blue}{56.86} / \textcolor{red}{80.39} & \textcolor{blue}{32.13} / \textcolor{red}{89} & \textcolor{blue}{32.5} / \textcolor{red}{89} & \textcolor{blue}{50.02} / \textcolor{red}{93.68} \\ \hline
Sui        & \textcolor{blue}{47.08} / \textcolor{red}{105.91} & \textcolor{blue}{32.58} / \textcolor{red}{65.13}  & \textcolor{blue}{48.57} / \textcolor{red}{97.14} & \textcolor{blue}{28.57}  / \textcolor{red}{35.71}  & \textcolor{blue}{46.72}/ \textcolor{red}{89.74} & \textcolor{blue}{45.53}  / \textcolor{red}{88.55} & \textcolor{blue}{50} / \textcolor{red}{97.46} \\ \hline \hline
\textbf{Average} & \textcolor{blue}{94.66} / \textcolor{red}{331.40} & \textcolor{blue}{65.34} / \textcolor{red}{183.08} & \textcolor{blue}{36.20} / \textcolor{red}{88.20} & \textcolor{blue}{40.39} / \textcolor{red}{55.76} & \textcolor{blue}{35.72} / \textcolor{red}{84.18} & \textcolor{blue}{35.61} / \textcolor{red}{84.4} & \textcolor{blue}{50.01} / \textcolor{red}{95.62} \\ \hline

\end{tabular}
\end{table*}

\subsection{Empirical Validation}
In this subsection, we use the validator set data outlined in Section~\ref{sec:empiricalanalysis} to recalculate weights using SRSW and LSW functions. We then compare decentralization metrics of the SRSW and LSW models against the linear stake weight model for specified validator sets. Our observations confirm the results of formal analysis. Our findings across the decentralization metrics are summarized as follows:

\begin{itemize}
    \item \textbf{Nakamoto Coefficients:} As shown in Figure~\ref{fig:nakamoto_coefficients}, both SRSW and LSW exhibit substantial improvements in Nakamoto coefficients compared to the linear model. Table~\ref{tab:percentimprovement} shows that \(\rho_{\mathbb{N}_L}\) increases ranged from 24.93\% to 169.87\% for SRSW and from 100\% to 785.71\% for LSW across different blockchains. Similar trends are observed for \(\rho_{\mathbb{N}_S}\), with average improvements of 65.34\% for SRSW and 183.01\% for LSW. These increases indicate stronger centralization resistance, as more validators are required to influence the consensus.
    \item \textbf{Gini Index:} The SRSW function yields a decrease in the Gini index, ranging from 16.25\% to 48.57\%, as shown in Figure~\ref{fig:gini} and Table~\ref{tab:percentimprovement}. For LSW, the improvement in Gini is even higher, ranging from 74.55\% to 97.14\%. This reduction reflects a more equitable distribution of stake weights, enhancing decentralization.
    \item \textbf{HHI:} While HHI values are already favorable due to the absence of monopoly in consensus weights, we observed additional improvements with SRSW and LSW. Specifically, SRSW achieved an average improvement of 40.39\%, while LSW yielded a 55.76\% improvement over the linear weight model.
    \item \textbf{Shapley Values:} Consistent with our empirical analysis, the Gini index of Shapley values for both liveness and safety closely follows the trends of the Gini index for validator weights, with a mean correlation of 0.96 between percentage improvements.
    \item \textbf{Zipf's Coefficient:} As shown in Figure~\ref{fig:zipfs}, substantial improvements are observed in the Zipf's law coefficient, particularly with LSW, indicating more decentralization.
\end{itemize}

\begin{figure*}[ht]
    \centering
    \begin{subfigure}[t]{0.5\textwidth}
        \centering
        \includegraphics[width=\linewidth, trim= {0.6cm 0.4cm 1cm 1.4cm}, clip]{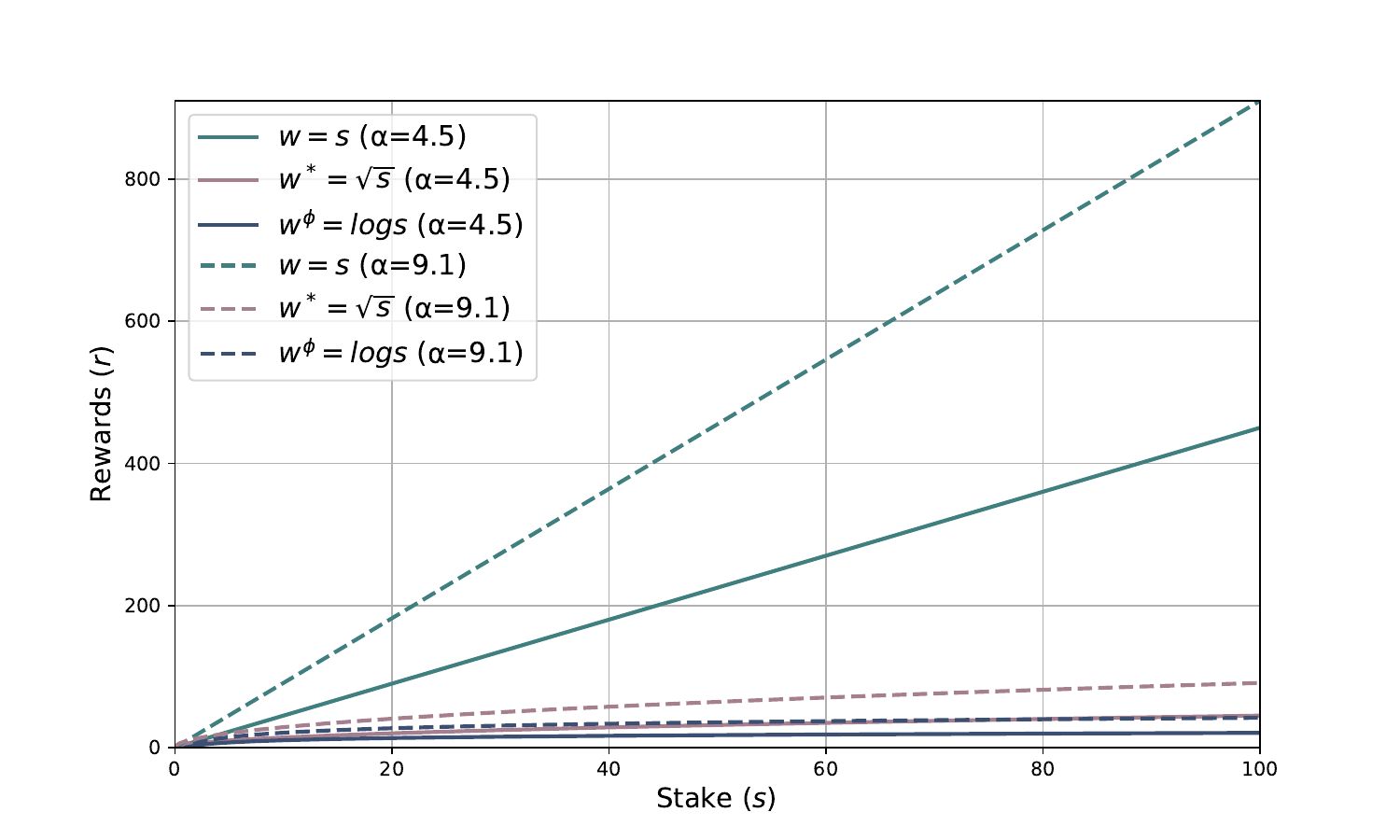}
        \caption{Reward rate with varying inflations}
        \label{fig:reward}
    \end{subfigure}
    \hspace{\fill}
    \begin{subfigure}[t]{0.48\textwidth}
        \centering
        \includegraphics[width=\linewidth, trim= {0cm 0cm 1cm 0cm}, clip]{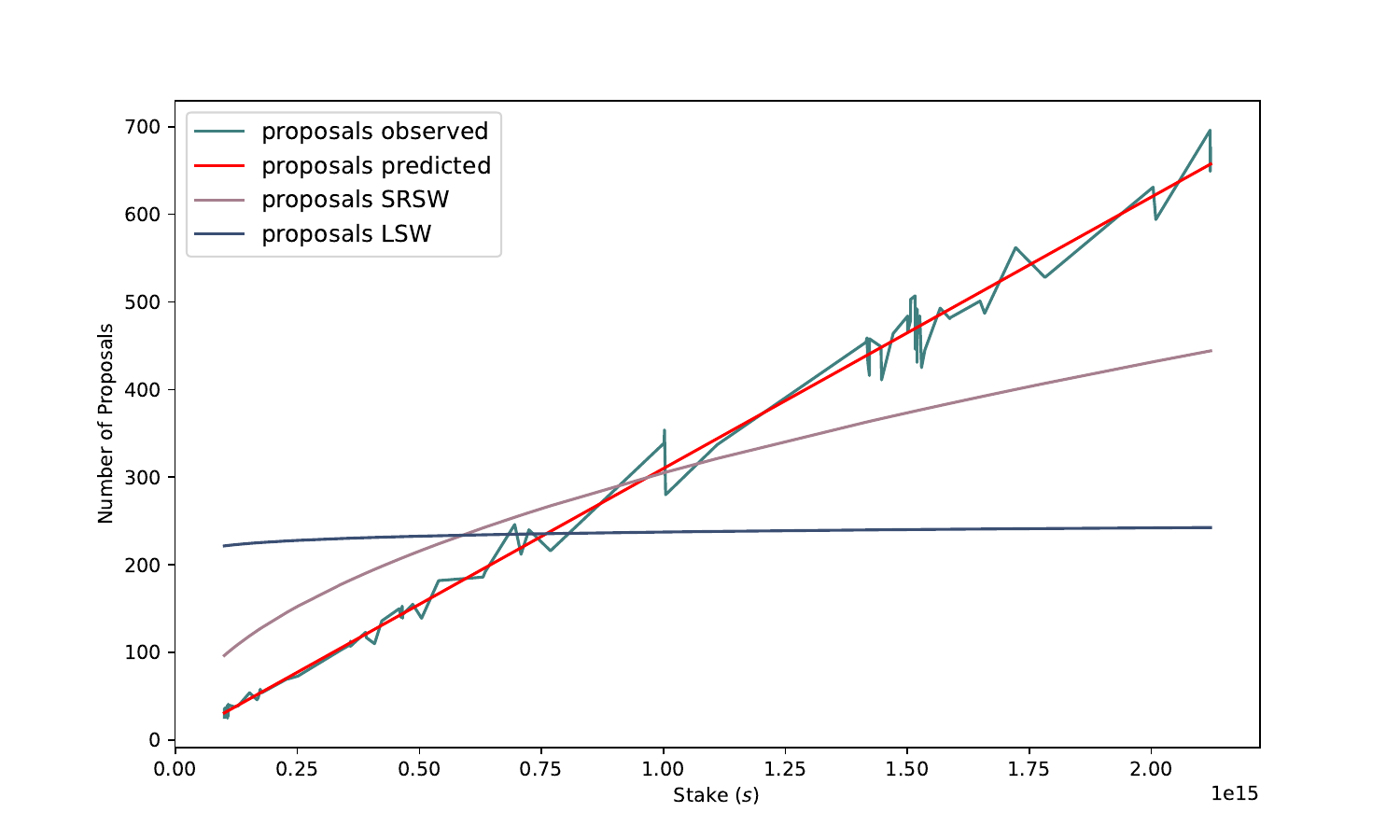}
        \caption{Block proposals made by validators}
        \label{fig:proposals}
    \end{subfigure}
    \caption{Implications of updated weight on block proposals and rewards}
    
    \label{fig:comparison}
\end{figure*}

\subsection{Secondary Effects of Improved Decentralization}
The observed improvements in decentralization metrics with SRSW and LSW models have notable secondary effects, as outlined below:

\begin{itemize}
    \item \textbf{Rewards - Rate of Growth:} 
    In Figure~\ref{fig:reward}, we analyze the growth rate of rewards using two annual inflation rates ($\alpha = \{4.5,9.1\}$), selected based on typical rates in blockchain implementations~\cite{stakingrewards}. This analysis underscores the benefits of the SRSW and LSW models, particularly in moderating reward growth for validators with larger stakes. By addressing the 'rich get richer' narrative, both models promote a fairer reward distribution, leading to more equitable weight compounding across epochs.

    \item \textbf{Block Generation Decentralization:} 
    Figure~\ref{fig:proposals} explores the dynamics of block proposal generation. Utilizing data from the Aptos blockchain~\cite{aptosBlockchain}, we initially demonstrate how the current block proposers are chosen based on linear stake weights. Subsequently, we compare this against the predicted block proposer distribution under the SRSW model. Our findings reveal that the SRSW and LSW models lead to a more diverse set of block proposers, playing a crucial role in mitigating Maximal Extractable Value (MEV) risks~\cite{daian2020flash} and enhancing censorship resistance~\cite{censorshipData}. This diversification in proposers is consistent with decentralization metrics discussed in contemporary studies~\cite{li2020comparison,li2023cross}, highlighting the SRSW and LSW models' contribution to decentralization.
\end{itemize}

\subsection{Comparing SRSW and LSW} 
Both SRSW and LSW outperform the linear stake weight model in decentralization, yet LSW consistently yields superior results. The stronger concavity of the logarithmic function in LSW reduces the influence of high-stake validators more effectively than the square root function in SRSW (see Theorem~\ref{theorem:lsw-superior}). Empirical evaluations show that LSW achieves a 60\% greater average improvement across all decentralization metrics compared to SRSW.

Nevertheless, SRSW has practical advantages, offering reduced computational complexity and serving as an efficient intermediary. Protocols may also consider non-linear alternatives, such as cubic root weighting, to align with specific decentralization goals.

\section{Related Work}
\label{sec:relatedwork}
Decentralization in blockchains, a cornerstone for blockchain efficacy, have been extensively explored, with a focus on governance~\cite{balajidecentralization, fritsch2022analyzing, kiayias2022sok, sharma2023unpacking, tan2023open}. Recent studies, such as those on token-based voting in DeFi protocols~\cite{messias2023understanding}, underscore the evolving complexities in blockchain governance. Particularly, decentralization research in PoS and DPoS systems have illuminated the challenges of weight concentration in these protocols~\cite{lin2021measuring, li2023liquid, liu2022understanding, li2023cross, li2020comparison, li2023hard, liu2022decentralization, kim2019stellar}.

These studies identify and quantify the weight concentration challenge in weighted consensus, yet solutions to this issue remain underexplored. Existing suggestions, such as capping proposals per validator~\cite{jeong2020centralized} and introducing reward sharing in validator staking pools~\cite{brunjes2020reward}, offer only partial remedies. Recent research on introducing virtual stake based on validator performance~\cite{mivsic2023towards} is addressing the challenge but raises potential vulnerabilities such as the `nothing-at-stake' problem~\cite{motepalli2021reward}. Our work diverges by enhancing decentralization directly at the consensus mechanism, without the need for new tokens or introducing associated vulnerabilities, thus improving decentralization more holistically.

\section{Conclusions}
\label{sec:conclusion}
This study introduced the Square Root Stake Weight (SRSW) and Log Stake Weight (LSW) functions to mitigate stake concentration in permissionless blockchains, achieving notable improvements in decentralization metrics, including the Gini index, Zipf's coefficient, and Nakamoto coefficients. While acknowledging the complexities of increasing Sybil resistance, we highlight the need for further investigation into practical implementations and governance frameworks. Future research could explore geospatial weight distributions and auction-based mechanisms for validator selection, offering pathways to enhance decentralization and economic efficiency in blockchain consensus.

\section*{Acknowledgment}
This work has been supported in part by NSERC, ORF, and the Sui Foundation.

\bibliographystyle{plain}
\bibliography{references}


\begin{IEEEbiography}
[{\includegraphics[width=1in,height=1.25in,clip,keepaspectratio]{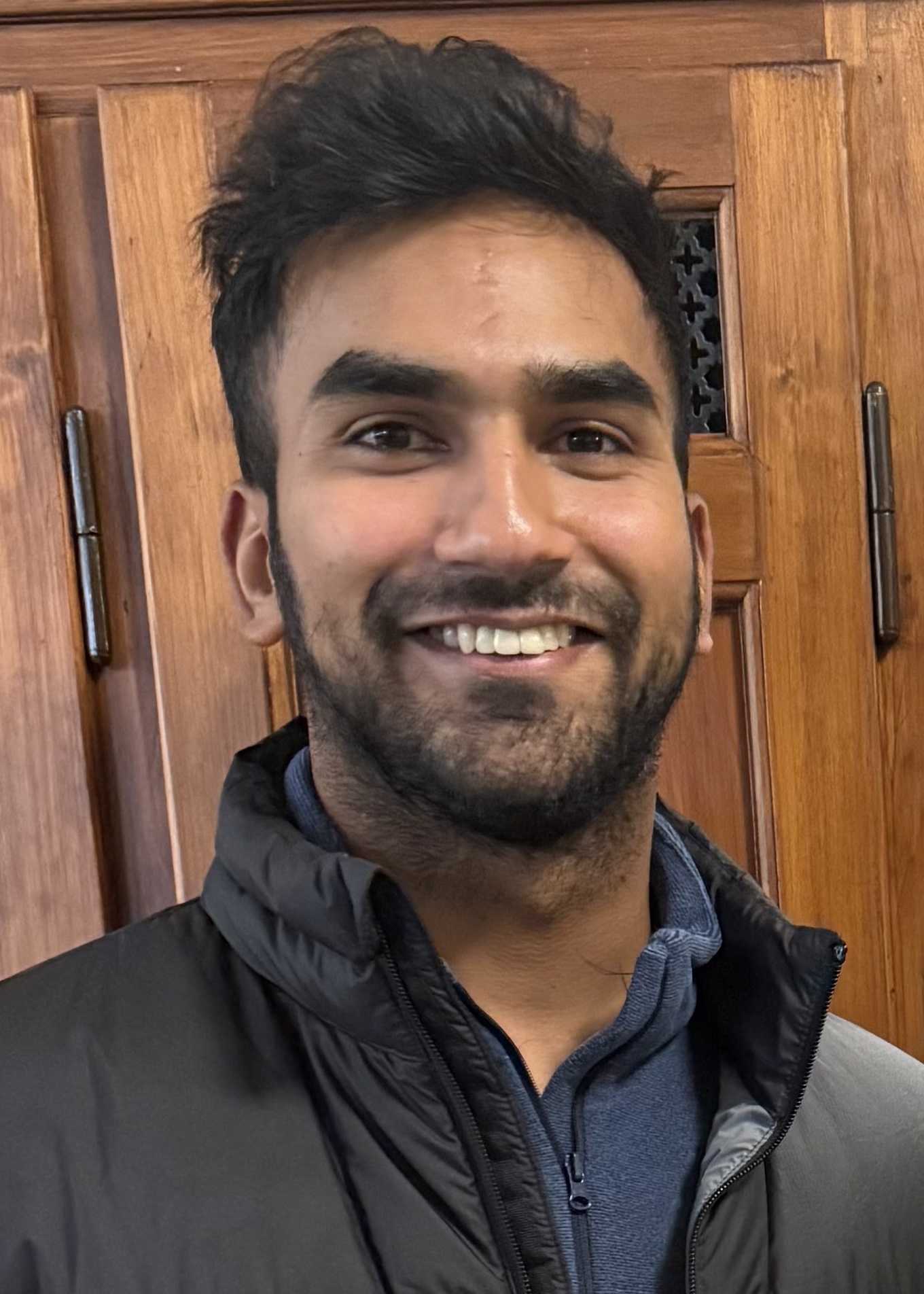}}]{Shashank Motepalli}
Motepalli is a PhD Candidate at the Sr. Rogers Department of Electrical and Computer Engineering, University of Toronto. His research focuses on decentralizing consensus mechanisms in blockchains, particularly mitigating stake concentration in Proof of Stake (PoS) systems and optimizing geospatial validator distribution for enhanced blockchain decentralization. He holds a Master’s degree in Information Technology from IIIT-Bangalore, India.
\end{IEEEbiography}
\vspace{-24pt}
\begin{IEEEbiography}
[{\includegraphics[width=1in,height=1.25in,clip,keepaspectratio]{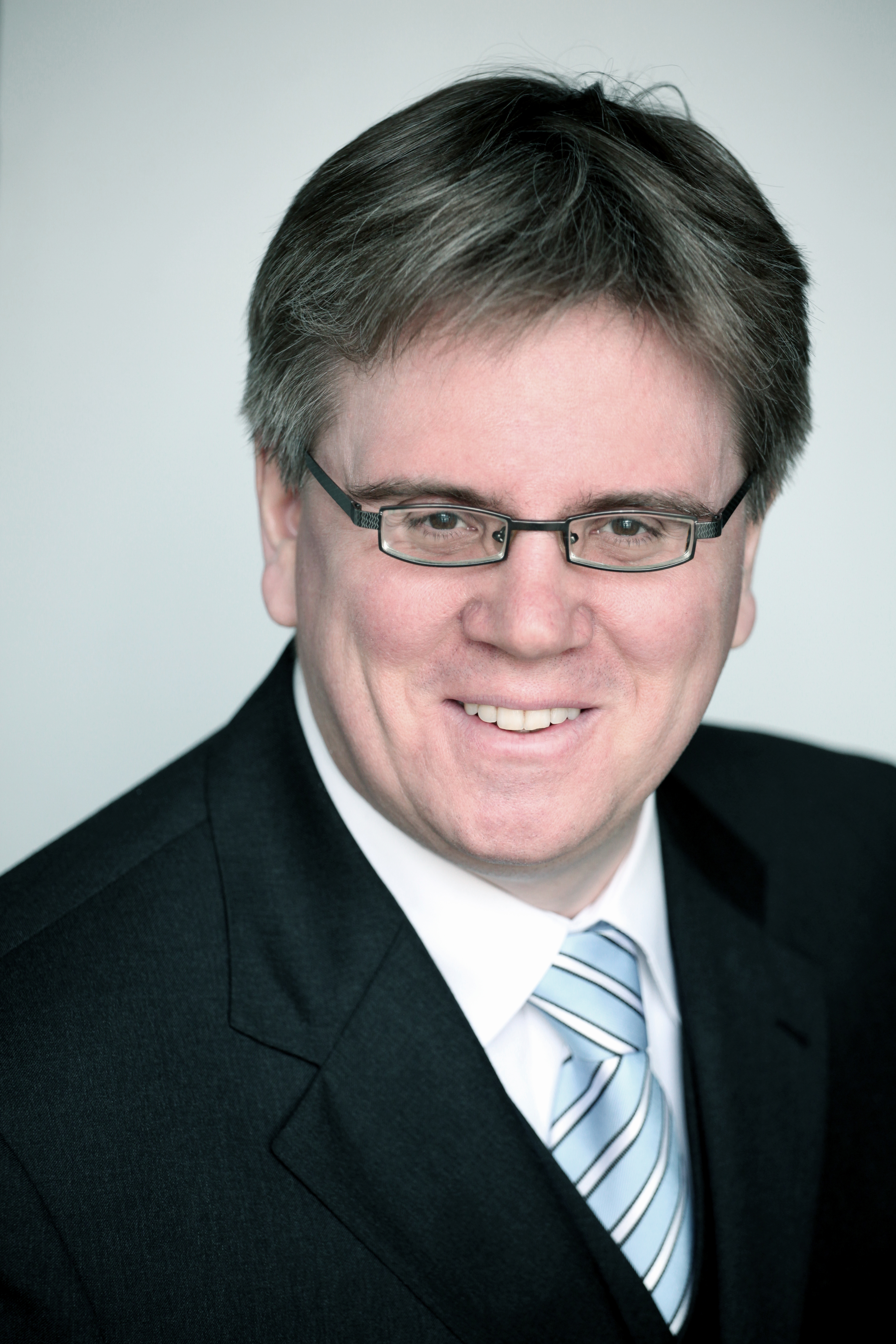}}]{Hans-Arno Jacobsen}
Professor Jacobsen holds the Jeffrey Skoll Chair in Computer Networking and Innovation at the Sr. Rogers Department of Electrical and Computer Engineering, University of Toronto, where he is a professor of Computer Engineering and Computer Science. His pioneering research lies at the intersection of distributed systems, data management and data science, with particular focus on blockchains, (complex) event processing, and cyber-physical systems. Over the past few years, he has increasingly become interested in quantum computing where, to this end, he is working on applications in molecular property prediction (computational chemistry) and quantum machine learning, in the long-term, aiming to endeavor into building distributed quantum computing abstractions. Professor Jacobsen is a Fellow of the IEEE.
\end{IEEEbiography}
\balance
\end{document}